\documentclass[twocolumn,sigconf,nonacm]{acmart}

\acmConference{}
\acmPrice{}
\acmISBN{}

\usepackage{etex,xy}
\xyoption{all}

\usepackage[colorinlistoftodos]{todonotes}

\newcommand{\E}{\mathbb E}

\newcommand{\Z}{\mathbb Z}
\newcommand{\N}{\mathbb N}
\newcommand{\Q}{\mathbb Q}
\newcommand{\K}{\mathbb K}
\newcommand{\A}{\mathbb A}

\newcommand{\ie}{i.e.,\ }
\newcommand{\dfield}[2]{({#1},{#2})}
\newcommand{\const}[2]{{\rm const}_{#2}{#1}} 
\newcommand{\mmod}{\operatorname{mod}}
\DeclareMathOperator*{\tprod}{{\textstyle\prod}}

\newcommand{\GG}{\mathbb G}
\newcommand{\F}{\mathbb F}
\newcommand{\HH}{\mathbb H}
\newcommand{\sigmaE}{$\Sigma$}
\newcommand{\piE}{$\Pi$}
\newcommand{\rE}{$R$}
\newcommand{\pisiE}{$\Pi\Sigma$}
\newcommand{\rpisiE}{$R\Pi\Sigma$}

\newcommand{\dfact}[3]{{#1}^{{#3},{#2}}}
\newcommand{\lr}[1]{\langle#1\rangle}
\newcommand{\vect}[1]{\boldsymbol{#1}}

\newcommand{\ssumB}[2]{\textstyle\sum\limits_{\scriptscriptstyle#1}^{\scriptscriptstyle#2}}
\newcommand{\sprodB}[2]{\textstyle\prod\limits_{\scriptscriptstyle#1}^{\scriptscriptstyle#2}}

\theoremstyle{plain}
\newtheorem{theorem}{Theorem}
\newtheorem{lemma}[theorem]{Lemma}
\newtheorem{corollary}[theorem]{Corollary}
\newtheorem{proposition}[theorem]{Proposition}

\theoremstyle{definition}
\newtheorem{definition}[theorem]{Definition}
\newtheorem{example}[theorem]{Example}

\theoremstyle{remark}
\newtheorem{remark}[theorem]{Remark}

\allowdisplaybreaks[4]

\begin{document}

\author{Jakob Ablinger}
\email{Jakob.Ablinger@risc.jku.at}
\affiliation{%
	\institution{Johannes Kepler University Linz, RISC}
	\streetaddress{Altenberger Str. 69}
	\city{Linz}
	\country{Austria}
	\postcode{4232}
}

\author{Carsten Schneider}
\email{Carsten.Schneider@risc.jku.at}
\affiliation{%
	\institution{Johannes Kepler University Linz, RISC}
	\streetaddress{Altenberger Str. 69}
	\city{Linz}
	\country{Austria}
	\postcode{4232}
}

\title[Solving linear difference equations with coefficients in rings with idempotent representations]{{\footnotesize\tt\vspace*{-1.3cm}
		RISC-Linz Report Series No. 21-04}\\[0.2cm]Solving linear difference equations with coefficients in rings with idempotent representations}\titlenote{Supported by the Austrian Science Foundation (FWF) grant SFB F50 (F5009-N15).}

\begin{abstract}
We introduce a general reduction strategy that enables one to search for solutions of parameterized linear difference equations in difference rings. Here we assume that the ring itself can be decomposed by a direct sum of integral domains (using idempotent elements) that enjoys certain technical features and that the coefficients of the difference equation are not degenerated. Using this mechanism we can reduce the problem to find solutions in a ring (with zero-divisors) to search solutions in several copies of integral domains. Utilizing existing solvers in this integral domain setting, we obtain a general solver where the components of the linear difference equations and the solutions can be taken from difference rings that are built e.g., by $R\Pi\Sigma$-extensions over $\Pi\Sigma$-fields. This class of difference rings contains, e.g., nested sums and products, products over roots of unity and nested sums defined over such objects.
\end{abstract}


\maketitle

\section{Introduction}

In the following we denote by $\dfield{\E}{\sigma}$ a \emph{difference ring (resp.\ field)}, this means that $\E$ is a ring (resp.\ field) $\E$ equipped with a ring (resp.\ field) automorphism $\sigma:\E\to\E$. We call $\dfield{\E}{\sigma}$ \emph{computable} if the basic operations of $\E$ and $\sigma$ are computable. We define the \emph{ring of constants} of $\dfield{\E}{\sigma}$ by
$\K=\const{\E}{\sigma}=\{c\in\E\mid \sigma(c)=c\}.$
By construction $\K$ will be a field, called the \emph{constant field} of $\dfield{\E}{\sigma}$.

Given such a difference ring $\dfield{\E}{\sigma}$ with a constant field $\K$, we are interested in the following problem: Given $\vect{a}=(a_0,\dots,a_m)\in\E^{m+1}$ and $\vect{f}=(f_1,\dots,f_d)\in\E^d$, find (if this is possible) a finite representation of all solutions $g\in\E$ and $c_1,\dots,c_d\in\K$ of the \emph{parameterized linear difference equation} (in short PLDE)
\begin{equation}\label{Equ:PLDE}
a_0\,g+a_1\,\sigma(g)+\dots+a_m\,\sigma^m(g)=c_1\,f_1+\dots+c_d\,f_d
\end{equation}
with coefficients $\vect{a}$ and parameters $\vect{f}$.
The solution set is defined by
\begin{equation*}
V=V(\vect{a},\vect{f},\E)=\{(c_1,\dots,c_d,g)\in\K^d\times\E\mid\text{ \eqref{Equ:PLDE} holds}\}
\end{equation*}
which forms a $\K$-subspace of $\K^d\times\E$. We say that \emph{we can compute all solutions in $\dfield{\E}{\sigma}$ of an explicitly given~\eqref{Equ:PLDE}} if $V$ is a finite dimensional vector space and one can compute a basis of $V$.
In particular, if $\E$ is an integral domain and $a_0\,a_m\neq0$, we have $\dim(V)\leq m+n$ by~\cite[Thm.~XII (page 272)]{Cohn:65}. In this case we say that \emph{we can solve (in general) parameterized linear difference equations in $\dfield{\A}{\sigma}$} if one can compute a basis of $V(\vect{a},\vect{f},\E)$ for any $\vect{0}\neq\vect{a}\in\E^{m+1}$ and $\vect{f}\in\E^d$. 

The problem to solve PLDEs (so far only in a field or integral domain $\E$) plays a central rule in symbolic summation and various algorithms. It covers as special cases the telescoping problem ($\vect{a}=(1,-1)$, $\vect{f}\in\E^1$) for, e.g., hypergeometric products~\cite{Gosper:78}, the creative telescoping problem ($\vect{a}=(1,-1)$ with appropriately chosen $\vect{f}\in\E^d$) for, e.g., hypergeometric products~\cite{Zeilberger:91}, or recurrence solving ($d=1$) for, e.g., rational or hypergeometric solutions~\cite{Abramov:89a,Petkov:92,vanHoeij:99}. The parameterized version is used also in holonomic summation~\cite{Chyzak:00} and generalizations of it~\cite{BRS:16}. Further details can found, e.g., in~\cite{Schneider:21}.

In particular, Karr's pioneering summation algorithm~\cite{Karr:81} established a highly general solver for first-order PLDEs in the setting of his \pisiE-field extensions (Def.~\ref{Def:PSFieldExt}). In this way, the coefficients $a_i$, parameters $f_i$ and the solutions $g$ can be given in a \pisiE-field $\dfield{\E}{\sigma}$ that is built formally by indefinite nested sums and products. Only recently, his general first-order solver has been pushed forward in~\cite{ABPS:21} to the higher-order case (including also a solver to find all hypergeometric solutions over $\E$), that covers most of the summation algorithms mentioned above as special cases. 

In this article we aim at further generalizations allowing in addition difference rings that are built by basic \rpisiE-ring extensions~\cite{DR1,DR3} (Def.~\ref{Def:APSExt}) where also products over roots of unity like $(-1)^n$ can arise. Based on the observation that such rings can be decomposed by a direct sum of integral domains using idempotent elements (which is one of the key tools in the Galois theory of difference equations~\cite{Singer:97,Singer:08}), we will develop in Section~\ref{Sec:IPDR} a general strategy to solve non-degenerated PLDEs in idempotent difference rings (Def.~\ref{Def:IdemPotentDR}). Inspired by~\cite{Salvy,Mallinger} we separate the potential solutions in their different components (Thm.~\ref{Subops}) and try to combine them accordingly to the full solution (Thm.~\ref{Thm:IdempotentSolver}). Utilizing this machinery, we will invoke in Section~\ref{Sec:GeneralSolvers} the general \pisiE-field solver~\cite{ABPS:21} (and variants of it) implemented within the summation package~\texttt{Sigma}~\cite{Schneider:07a} to derive various new algorithms (see Theorems~\ref{Thm:QuotientRPSRingSolver} and~\ref{Thm:SpecialRPiSiSolver}) in order to solve non-degenerated PLDEs in basic \rpisiE-rings defined over \pisiE-field-extensions. As a special case, the ground field can be, e.g., the mixed multibasic difference field~\cite{Bauer:99} introduced in Remark~\ref{Remark:MixedCase}. 
After a concrete example in Section~\ref{Sec:Example} we conclude with Section~\ref{Sec:Conclusion}.

\newpage

\section{$\!\!\!$PLDEs in idempotent difference rings}\label{Sec:IPDR}

It will be convenient to denote by $s\mmod \lambda$ with $s\in\Z$ the unique value $l\in\{0,\dots,\lambda-1\}$ with $\lambda\mid s-l$.
\begin{definition}\label{Def:IdemPotentDR}
 Let $\dfield{\E}{\sigma}$ be a difference ring and let $e_s\in\E$ with $0\leq s<\lambda$
 be elements such that
 \begin{itemize}
  \item they are idempotent (i.e., $e_s^2=e_s$),
  \item pairwise orthogonal (i.e., $e_s e_t=0$ if $s\neq t$),
  \item and $\sigma(e_s)=e_{s+1\mmod \lambda}.$
 \end{itemize}
If $\dfield{\E}{\sigma}$ can be decomposed in the form 
\begin{equation}\label{Equ:decomposition}
\E=e_0\,\E\oplus e_2\,\E\oplus\dots\oplus e_{\lambda-1}\,\E
\end{equation}
such that $e_i\E$ forms an computable integral domain, then $\dfield{\E}{\sigma}$ is called an \textit{idempotent difference ring of order $\lambda$}.
\end{definition}
\noindent Note that, if $\dfield{\E}{\sigma}$ is an idempotent difference ring of order $\lambda$ then $\dfield{e_s\,\E}{\sigma^\lambda}$ is a difference ring and $\sigma$ is a difference ring isomorphism\footnote{A difference isomorphism $\tau:\A_1\to\A_2$ between two difference rings $\dfield{\A_i}{\sigma_i}$ with $i=1,2$ is a ring isomorphism with $\tau(\sigma_1(f))=\sigma_2(\tau(f))$ for all $f\in\A_1$.} between $\dfield{e_s\,\E}{\sigma^\lambda}$ and $\dfield{e_{s+1\mmod \lambda}\,\E}{\sigma^\lambda}$.
\begin{lemma}\label{Lemma:componentshifts}
 Let $\dfield{\E}{\sigma}$ be an idempotent difference ring of order $\lambda$ and let $g=\sum_{s=0}^{\lambda-1}e_s\,g_s\in\E$, then applying $\sigma$ means that the component $e_sg_s$ is moved cyclically to $\dfield{e_{{s+1}\mmod \lambda}\,\E}{\sigma^\lambda}$
\end{lemma}
\begin{proof}
Fix $s$ with $0\leq s<\lambda$, since $g_s\in\E$ we can write $g_s=\sum_{i=0}^{\lambda-1}e_i\,h_i$ for some $h_i\in\E$. Now applying $\sigma$ to $e_sg_s$ gives:
\begin{align*}
 \sigma(e_sg_s)&=\sigma\left(e_s\sum_{i=0}^{\lambda-1}e_i\,h_i\right)=\sigma\left(e_s\right)\sum_{i=0}^{\lambda-1}\sigma\left(e_i\right)\sigma\left(h_i\right)\\
 &=e_{s+1\mmod \lambda}\sum_{i=0}^{\lambda-1}e_{i+1\mmod \lambda}\sigma\left(h_i\right)=e_{s+1\mmod \lambda}\sigma\left(h_s\right).
\end{align*}
Since $\sigma(h_s)\in\E$ we have that $\sigma(e_sg_s)\in e_{s+1\mmod \lambda}\E$.
\end{proof}
\noindent For an idempotent difference ring $\dfield{\E}{\sigma}$ of order $\lambda$, with idempotent elements $e_s\in\E$ with $0\leq s<\lambda$ the structure given by Lemma~\ref{Lemma:componentshifts} can be illustrated as follows:
\begin{align*}
\xymatrix@!C=0.01cm{
\E\,= \,e_0\,\E\ar@/^1pc/^{\sigma}[rr] \quad & \oplus& e_1\,\E\ar@/^1pc/^{\sigma}[rr]&\oplus&\dots\ar@/^1pc/^{\sigma}[rr]&\oplus& e_{\lambda-2}\,\E\ar@/^1 pc/^{\sigma}[rr]&
\oplus& e_{\lambda-1}\,\E\ar@/^{1.3pc}/^{\sigma}[llllllll]}.
\end{align*}
The following lemma is immediate.
\begin{lemma}\label{Lemma:cyclicshift}
Let $\dfield{\E}{\sigma}$ be an idempotent difference ring of order $\lambda$ and let $g=\sum_{s=0}^{\lambda-1}e_s\,g_s\in\E$ and $j\in\N$ then 
\begin{align}
\sigma^j(g)=\sum_{s=0}^{\lambda-1}e_{s+j\mmod \lambda}\sigma^j(g_{s})=\sum_{s=0}^{\lambda-1}e_s\sigma^j(g_{s-j\mmod \lambda}).
\end{align}
\end{lemma}
\begin{definition}
Let $\dfield{\E}{\sigma}$ be an idempotent difference ring of order $\lambda$ with idempotent elements $e_s\in\E$ with $0\leq s<\lambda$. Then
 $\pi:\E\to\E$ with $\pi(g)\mapsto g_0$ where $g=\sum_{s=0}^{\lambda-1} e_s g_s$ is called a \textit{projection}.
\end{definition}
\noindent In this article we will always consider the projection on the first component, however each projection to an arbitrary component would do the job. The following lemma summarizes several properties of the projection. 
\begin{lemma}\label{Lemma:pi_props}
Let $\dfield{\E}{\sigma}$ be an idempotent difference ring of order $\lambda$ with idempotent elements $e_s\in\E$ with $0\leq s<\lambda$ and let  $\pi:\E\to\E$ be a projection. For $g,h\in\E$ we have
\begin{equation}
\pi(g+h)=\pi(g)+\pi(h)\quad \text{ and } \quad \pi(g\cdot h)=\pi(g)\cdot\pi(h).
\end{equation}
In addition, for $j\in \N$ and $0\leq s<\lambda$ we have
\begin{equation}
\pi(\sigma^j(e_s))=
    \begin{cases}
        1&\text{if }s+j= 0 \pmod \lambda\\
        0&\text{if }s+j\neq 0 \pmod \lambda,
    \end{cases}
\end{equation}
and for $j\in \N$ and $g=\sum_{s=0}^{\lambda-1} e_s g_s$ we have
\begin{equation}
\pi(g)=e_0 g\quad \text{ and } \quad \pi(\sigma^j(g))= \sigma^j(g_{-j\mmod \lambda})
\end{equation}
\end{lemma}
\begin{proof}
 Let $g=\sum_{s=0}^{\lambda-1}e_s\,g_s\in\E$ and $h=\sum_{s=0}^{\lambda-1}e_s\,h_s\in\E$ then $g+h=\sum_{s=0}^{\lambda-1}e_s\,(g_s+h_s)\in\E$ and hence $\pi(g+h)=g_0+h_0=\pi(g)+\pi(h)$. Similarly, since $g\cdot h=\sum_{s=0}^{\lambda-1}e_s\,(g_s\cdot h_s)\in\E$ we have $\pi(g\cdot h)=g_0\cdot h_0=\pi(g)\cdot\pi(h)$. For $j\in\N,0\leq s<\lambda$ we have that $\sigma^j(e_s)=e_{s+j\mmod \lambda}$, hence $\pi(\sigma^j(e_s))=\pi(e_{s+j\mmod \lambda})$ which clearly evaluates to 1 if $s+j=0\pmod \lambda$ and to 0 if $s+j\neq 0\pmod \lambda$. Finally, from Lemma~\ref{Lemma:cyclicshift} we know that $\sigma^j(g)=\sum_{s=0}^{\lambda-1}e_s\sigma^j(g_{s-j\mmod \lambda})$, hence $\pi(\sigma^j(g))=\sum_{s=0}^{\lambda-1}\pi(e_s)\pi(\sigma^j(g_{s-j\mmod \lambda}))=\pi(\sigma^j(g_{-j\mmod \lambda}))$. Since $\pi(\sigma^j(g_{-j\mmod \lambda}))=e_0\sigma^j(g_{-j\mmod \lambda})=\sigma^j(e_{-j\mmod \lambda}g_{-j\mmod \lambda})$ we have that
 $\sigma^j(g)=\sigma^j(g_{-j\mmod \lambda})$.
\end{proof}
\begin{definition}
Let $\dfield{\E}{\sigma}$ be an idempotent difference ring of order $\lambda$ and let $\pi:\E\to\E$ be a projection. For $\vect{a}=(a_0,a_1,\dots,a_m)\in\E^{m+1}$
we define the $(m+1)\lambda-m \times (m+1)\lambda$ \textit{shift projection} matrix by 
\begin{flalign*}
M_{\sigma,\pi}(\vect{a}):=&&
\end{flalign*}
{\tiny
\begin{align*}
    &\begin{pmatrix} 
    \pi(p_0) &         \pi(p_1) & \cdots &   \pi(p_m) &    0 & 0 &  \cdots &                       0 \\
    0        & \pi(\sigma(p_0)) & \cdots & \pi(\sigma(p_{m-1})) &   \pi(\sigma(p_m)) & 0 &  \cdots & 0 \\
    \vdots&  & \ddots\\
    0        &   0 & \cdots &    0 & \pi(\sigma^{k}(p_0)) &   &  \cdots & \pi(\sigma^{k}(p_m)) 
    \end{pmatrix},
\end{align*}
}%
where $k:=(m+1)\lambda-m-1.$
\end{definition}
\begin{definition}\label{Def:NonDegeneratedOp}
Let $\dfield{\E}{\sigma}$ be an idempotent difference ring of order $\lambda$ and let $\pi:\E\to\E$ be a projection. A vector $\vect{a}=(a_0,a_1,\dots,a_m)\in\E^{m+1}$ is called \textit{non-degenerate} if the shift projection matrix $M_{\sigma,\pi}(\vect{a})$ has full rank, \ie the rows are linearly independent.
Likewise, a  linear difference operator $\sum_{i=0}^m a_i\sigma^i\in\E[\sigma]$ with $a_i\in\E$ is called \textit{non-degenerate} if $\vect{a}$ is \textit{non-degenerate}.
\end{definition}
\noindent Note, that for instance a linear difference operator $L=\sum_{i=0}^m a_i\sigma^i\in\E[\sigma]$ that is a multiple of an idempotent element $e_i$ \ie $e_i\mid a_i$ for all $0\leq i\leq m$ is not non-degenerate, since for such an operator the shift projection matrix would contain a zero row. Similarly, $L$ for which all coefficients vanish for a certain component is as well degenerate, since for such an operator the shift projection matrix would contain $m+1$ zero columns, see Example \ref{Ex:Degenerate} below.\\
In the following lemma, we state an immediate criterion which implies that a linear difference operator is \textit{non-degenerate}.
\begin{lemma}
 Let $\dfield{\E}{\sigma}$ be an idempotent difference ring of order $\lambda$ and let $\pi:\E\to\E$ be a projection. A linear difference operator 
$L:=\sum_{i=0}^mp_i\sigma^i\in\E[\sigma]$, with $a_i\in \E$,
is \textit{non-degenerate} if either $a_m$ or $a_0$ is a unit in $\E$.
\end{lemma}
\noindent Given a non-degenerate linear difference operator, the following theorem shows, that it is possible to define non-zero linear difference operators for each component. It is inspired by \cite{Salvy,Mallinger}. 

\begin{theorem}\label{Subops}
Let $\dfield{\E}{\sigma}$ be an idempotent difference ring of order $\lambda$ with idempotent elements $e_s\in\E$ with $0\leq s<\lambda$, let  $\pi:\E\to\E$ be a projection and let $\vect{a}=(a_0,\dots,a_{m})\in\E^{m+1}$ with $a_m\neq 0$ be non-degenerated. Consider the linear difference equation 
\begin{align}\label{Equ:lde}
\sum_{i=0}^ma_i\sigma^i(g)=\varphi.
\end{align}
with $\varphi\in \E$, which is satisfied by $g=\sum_{s=0}^{\lambda-1}e_s\,g_s\in\E$ and let $k\in \N$ with $0 \leq k<\lambda$. Then there exist
$b_{k,i}\in e_k\E$, not all zero, and $\varphi_k\in e_k\E$ such that
\begin{align}\label{Equ:sublde}
\sum_{i=0}^mb_{k,i}(\sigma^\lambda)^i(g_k)=\varphi_k.
\end{align}
If $\dfield{\E}{\sigma}$ is computable, then the $b_{k,i}$ and $\varphi_k$ can be computed. 
\end{theorem}
\begin{proof}
From~\eqref{Equ:lde} we can deduce for $j\in\N$ that 
\begin{align}\label{Equ:recreplaced}
\sigma^j\left(\sum_{i=0}^m a_i \sigma^i(g)\right)=\sigma^j \left(\varphi \right)
\end{align}
or equivalently 
\begin{align*}
\sum_{i=0}^m \sigma^j(a_i) \left(\sigma^{i+j}(e_0) \sigma^{i+j}(g_0) + \cdots\right.& \\
\left.\dots+ \sigma^{i+j}(e_{\lambda-1})\sigma^{i+j}(g_{\lambda-1})\right)&=\sigma^j \left(\varphi \right).
\end{align*}
Applying the projection $\pi$ and using Lemma~\ref{Lemma:pi_props} yields
\begin{align*}
\sum_{i=0}^m \pi(\sigma^j(a_i)) \pi(\sigma^{i+j}(g_{-(i+j)\mmod{\lambda}}))=\pi\left(\sigma^j \left(\varphi \right)\right),
\end{align*}
since for $1\leq l<\lambda$,
\begin{equation*}
\pi(\sigma^{i+j}(e_l))=
    \begin{cases}
        1&\text{if }l= -(i+j) \pmod \lambda\\
        0&\text{if }l\neq -(i+j) \pmod \lambda.
    \end{cases}
\end{equation*}
Now, by Lemma~\ref{Lemma:cyclicshift} and Lemma~\ref{Lemma:pi_props} we find
\begin{align}\label{Equ:equations}
\sum_{i=0}^m \pi(\sigma^j(a_i)) \sigma^{i+j}(g_{-(i+j)\mmod{\lambda}})=\pi\left(\sigma^j \left(\varphi \right)\right).
\end{align}
Now, plugging in $j=0,1,2,\ldots, (m+1)\lambda-m-1$ into~\eqref{Equ:equations} yields the linear system
\begin{align}\label{Equ:LinearSystem}
 M_{\sigma,\pi}(\vect{a})\cdot
\begin{pmatrix}
                    \sigma^0(g_{0\mmod{\lambda}})\\
                    \sigma^1(g_{-1\mmod{\lambda}})\\
                    \sigma^2(g_{-2\mmod{\lambda}})\\
                    \vdots\\
                    \sigma^{\nu}(g_{-k\mmod{\lambda}})\\
\end{pmatrix}=
                    \begin{pmatrix}
                    \pi(\sigma^0(\varphi))\\
                    \pi(\sigma^1(\varphi))\\
                    \pi(\sigma^2(\varphi))\\
                    \vdots\\
                    \pi(\sigma^{\nu}(\varphi))\\
                    \end{pmatrix},
\end{align}
where $\nu:=(m+1)\lambda-m-1$.
Since $\vect{a}$ is non-degenerate and hence $M_{\sigma,\pi}(\vect{a})$ has full rank, we can solve this system in terms of $m$ variables. Finally, we can plug this solution into~\eqref{Equ:sublde}. Since this leads to a linear system of at most $m+1$ equations in $m+2$ variables, which has a nontrivial solution, we can determine the coefficients $b_{k,i}$ and $\varphi_k$ of~\eqref{Equ:sublde}. In particular, if $\E$ is computable, the $b_{k,i}$ and $\varphi_k$ can be computed.
\end{proof}
%

\begin{remark}\label{Remark:PropertyOfRHS}
Let $\dfield{\E}{\sigma}$ be a field extension of a difference ring $\dfield{\A}{\sigma'}$, i.e., $\A$ is a subring of $\E$ and $\sigma|_{\A}=\sigma'$, and suppose that the $\vect{a}\in\A^{m+1}$ and $\phi\in\E$. 
Then, since we plug solutions of the linear system~\eqref{Equ:LinearSystem} into~\eqref{Equ:sublde}, the right-hand sides in~\eqref{Equ:sublde} have the form
$$\varphi_k=\sum_{l=0}^sf_l\,\pi(\sigma^{l}(\varphi))$$
with $f_0,\dots,f_s\in\A$ for some $s\in\N$.
\end{remark}

\begin{example}
Consider the idempotent difference ring $\dfield{\Q(x)[y]}{\sigma}$ with $\sigma(x)=x+1$ and $\sigma(y)=-y$ and the idempotent elements $e_0=\frac{1-y}2$ and $e_1=\frac{1-y}2$. Let $\vect{a}=(x,x,1,y),$ then the shift projection matrix $M_{\sigma,\pi}(\vect{a})$ yields
$$
\left(
\begin{array}{cccccccc}
 x & x & 1 & -1 & 0 & 0 & 0 & 0 \\
 0 & 1+x & 1+x & 1 & 1 & 0 & 0 & 0 \\
 0 & 0 & 2+x & 2+x & 1 & -1 & 0 & 0 \\
 0 & 0 & 0 & 3+x & 3+x & 1 & 1 & 0 \\
 0 & 0 & 0 & 0 & 4+x & 4+x & 1 & -1 \\
\end{array}
\right),
$$
which has full rank. If $g=e_0 g_0+e_1 g_1\in\E$ is a solution of 
$$xg+x\sigma(g)+\sigma^2(g)+y\sigma^3(g)=0$$ then we find for $g_0$ and $g_1$:
\begin{align*}
 x (1+x) (5+2 x)g_0+(7+7 x-3 x^2-2 x^3)\sigma^2(g_0)&\\
 +4 (1+x)(\sigma^2)^2(g_0)+(1+2 x)(\sigma^2)^3(g_0)&=0,\\
  x (1+x)g_1+(3+x-x^2)\sigma^2(g_1)-2(\sigma^2)^2(g_1)+(\sigma^2)^3(g_1)&=0.
\end{align*}
\end{example}
\noindent  Note that even in the degenerated case it might be possible to use the method stated in the proof of Theorem~\ref{Subops} to construct non-zero linear difference equations for some of the components.
\begin{example}\label{Ex:Degenerate}
Again we consider the idempotent difference ring $\dfield{\Q(x)[y]}{\sigma}$ with $\sigma(x)=x+1$ and $\sigma(y)=-y$ and the idempotent elements $e_0=\frac{1-y}2$ and $e_1=\frac{1-y}2$. Let $\vect{a}=(y-1,x(y+1),y-1,x(y+1))$, then the shift projection matrix $M_{\sigma,\pi}(\vect{a})$ yields
$$
\left(
\begin{array}{cccccccc}
 -2 & 0 & -2 & 0 & 0 & 0 & 0 & 0 \\
 0 & 0 & 2 (1+x) & 0 & 2 (1+x) & 0 & 0 & 0 \\
 0 & 0 & -2 & 0 & -2 & 0 & 0 & 0 \\
 0 & 0 & 0 & 0 & 2 (3+x) & 0 & 2 (3+x) & 0 \\
 0 & 0 & 0 & 0 & -2 & 0 & -2 & 0 \\
\end{array}
\right),
$$
which clearly doesn't have full rank. Still if $g=e_0 g_0+e_1 g_1\in\E$ is a solution of 
$$(y-1)g+x(y+1)\sigma(g)+(y-1)\sigma^2(g)+x(y+1)\sigma^3(g)=0$$ 
then the first component $g_0$ satisfies
$g_0+\sigma^2(g_0)=0$ but we do not find a non-trivial linear difference equation for $g_1$.
\end{example}

 With this notion 
the following corollary is immediate.
\begin{corollary}\label{Cor:GetParaRecsForCopies}
Let $\dfield{\E}{\sigma}$ and $\vect{a}\in\E^{m+1}$ be as stated in Theorem~\ref{Subops}.
Consider the PLDE~\eqref{Equ:PLDE}
with $f_i\in \E$ and $c_i\in\K$, which is satisfied by $g=\sum_{s=0}^{\lambda-1}e_s\,g_s\in\E$ and let $k\in \N$ with $0 \leq k<\lambda$. Then there exist 
$b_{k,i}\in e_k\E$, not all zero, and $f_{k,j}\in e_k\E$ such that
\begin{align}\label{Equ:ParaRecsForCopy}
\sum_{i=0}^mb_{k,i}(\sigma^\lambda)^i(g_k)=c_1\,f_{k,1}+\dots+c_d\,f_{k,d}.
\end{align}
In particular, if $\dfield{\E}{\sigma}$ is computable, the $a_{k,i}$ and $f_{k,j}$ are computable.
\end{corollary}

We are now ready to obtain a general strategy to solve PLDEs under the assumption that one can solve PLDEs in $\dfield{e_0\,\E}{\sigma^{\lambda}}$. Note that the task to compute for $f_0,\dots,f_d\in e_k\,\E$ a basis of
\begin{equation}\label{Equ:AnnProblem}
\{(c_1,\dots,c_d)\in\K^d\mid c_1\,f_1+\dots+c_d\,f_d=0\}
\end{equation}
is a special case by setting $g=0$ in~\eqref{Equ:PLDE}.

\begin{theorem}\label{Thm:IdempotentSolver}
	Let $\dfield{\E}{\sigma}$ be an idempotent difference ring with the idempotent elements $e_0,\dots,e_{\lambda-1}$ and constant field $\K$, and let $\vect{a}\in\E^{m+1}$ and $\vect{f}\in\E^d$. If $\const{e_0\,\E}{\sigma^{\lambda}}=e_0\,\K$ and $\vect{a}$ is non-degenerated, $V(\vect{a},\vect{f},\E)$ has a finite basis. If $\dfield{\E}{\sigma}$ is computable and PLDEs in $\dfield{e_0\E}{\sigma^{\lambda}}$ can be computed, a basis of $V(\vect{a},\vect{f},\E)$ can be computed.
\end{theorem}
\begin{proof}
	We look for a basis of $V=V(\vect{a},\vect{f},\E)$ over $\K$ for a non-degenerated $\vect{a}\in\E^{m+1}$ and $\vect{f}\in\E^d$.
	By Corollary~\ref{Cor:GetParaRecsForCopies} there exist $b_{k,i}\in e_k\,\E$, not all zero, and $f_{k,j}\in e_k\,\E$ with~\eqref{Equ:ParaRecsForCopy}. Since $e_k\,\E$ for $0\leq k<\lambda$ are integral domains, we can take a finite basis $\{(e_kc^{(k)}_{j,1},\dots,e_kc^{(k)}_{j,d},e_k\gamma^{(k)}_j)\}_{1\leq j\leq \delta_k}\subseteq(e_k\,\K)^d\times(e_k\,\E)$ with $c^{(k)}_{j,l}\in\K$ of $V_k=V((b_{k,0},\dots,b_{k,m}),(f_{k,1},\dots,f_{k,d}),e_k\,\E)$ over $e_k\,\K$. If $\delta_k=0$ for some $0\leq k<\lambda$ it follows that $V=\{\vect{0}\}$ and we get the empty basis. 
	Otherwise, we can take a basis of 
	\begin{multline*}
	W=\{(c_1,\dots,c_d,e_0 g_0+\dots+e_{\lambda-1}g_{\lambda-1})\in\K^d\times\E\mid\\ 
	(e_kc_1,\dots,e_kc_d,e_kg_k)\in V_k\text{ for $0\leq k<\lambda$}\}.
	\end{multline*}
	as follows. We define $C_k=(c_{j,l})_{1\leq j\leq \delta_k,1\leq l\leq d}$ for $0\leq k<\lambda$ and take a $\K$-basis, say 
	\begin{multline*}
	\{(d_{l,0,1},\dots,d_{l,0,\delta_1},\dots,d_{l,\lambda-1,1},
	\dots,d_{l,\lambda-1,\delta_{\lambda-1}})\}_{1\leq l\leq r},
	\end{multline*} 
	of the $\K$-vector space
	\begin{multline*}
	\{(d_{0,1},\dots,d_{0,\delta_1},\dots,d_{\lambda-1,1},\dots,d_{\lambda-1,\delta_{\lambda-1}})\in\K^{\delta_0+\dots+\delta_{\lambda-1}}\mid\\ (d_{0,1},\dots,d_{0,\delta_{0}})C_0=\dots=(d_{\lambda-1,1},\dots,d_{\lambda-1,\delta_{\lambda-1}})C_{\lambda-1}\}.
	\end{multline*}
	$\bullet$ If $r>0$, we proceed as follows. We define 
	for $1\leq l\leq r$ the elements 
	$$g_l=g^{(0)}_l+\dots+g^{(\lambda-1)}_l\in\E$$
	with
	$g^{(k)}_l=d_{l,k,1}\,e_{k}\,\gamma^{(k)}_1+\dots+d_{l,k,\delta_k}\,e_{k}\,\gamma^{(k)}_{\delta_k}$ where $0\leq k<\lambda$, and define for $1\leq l\leq r$ the constants
	$$(c_{l,1},\dots,c_{l,d})=(d_{l,0,1},\dots,d_{l,0,\delta_{1}})C_1\in\K^d.$$
	Then $B=\{(c_{l,1},\dots,c_{l,d},g_l)\}_{1\leq l\leq r}$ forms a bases of $W$. 
	Now we plug in the found basis elements into~\eqref{Equ:PLDE} and obtain linear constraints. Fulfilling them by combining the basis elements accordingly will lead finally to a basis of the solution space $V$.
	For this final step, take $C=(c_{l,i})_{1\leq l\leq r,1\leq i\leq d}$ with $c_{l,i}\in\K$ and $\vect{g}=(g_1,\dots,g_r)\in\E^r$, and define
	$$\vect{f}':=C\vect{f}^t-(a_m\sigma^m(\vect{g})+\dots+a_0\vect{g})\in\E^r;$$
	here applying $\sigma$ to a vector means to apply $\sigma$ to each component. Note that nonzero elements in $\vect{f'}$ reflect the disagreement of the so far found basis $B$ to be also a basis of $V$.  
	To complete the construction, we compute for the vector space
	\begin{equation}\label{Equ:CriticalBases}
	W'=\{(\kappa_1,\dots,\kappa_r)\in\K^r\mid (\kappa_1,\dots,\kappa_r)\vect{f}'\}
	\end{equation}
	the basis $\{(\kappa_{i,1},\dots,\kappa_{i,r})_{1\leq i\leq s}\subseteq\K^s$; here one collects the components of $\vect{f}'$ w.r.t.\ the $e_k$ for $0\leq k<\lambda$ (which is justified since $e_0,\dots,e_{\lambda-1}$ are linearly independent), derives the bases in the integral domains $e_k\,\E$ for each $0\leq k<\lambda$ and computes the intersection of the corresponding vector spaces to get a basis of $W'$. 
	If $s=0$, $V=\{\vect{0}\}$ and we get the empty basis of $V$. Otherwise,
	take $D=(\kappa_{i,j})_{1\leq i\leq s, 1\leq j\leq r}$ and define the entries of the matrix $(c'_{i,j})_{1\leq i\leq s,1\leq j\leq d}:=D\,C$ and the entries of the vector $(g'_1,\dots,g'_s):=D\,(g_1,\dots,g_r)\in\E^s$. By construction  $\{(c'_{i,1},\dots,c'_{i,d},g'_i)\}_{1\leq i\leq s}\subseteq\K^d\times\E$ is a basis of $V$.\\
	$\bullet$ If $r=0$, it follows that $V\subset\{0\}^d\times\E$, i.e., we only have to search for homogeneous solutions of~\eqref{Equ:PLDE}. Using the above construction we get a basis of the form $\{(0,g'_i)\}_{1\leq i\leq s}\cup\{(1,0)\}$ of $V(\vect{a},(0),\E)$. This gives the basis $\{(0,\dots,0,g'_i)\}_{1\leq i\leq s}\subseteq\{0\}^d\times\E$ of $V$.\\
We observe that the construction above can be carried out explicitly if the algorithmic assumptions hold: First, we can compute the bases of $V_i$; more precisely, we move the problem with the isomorphism $\sigma^{\lambda-i}$ to the zero component, solve it there and move it back with $\sigma^{i}$. Further, we can solve the various linear algebra problems in $\K$. Finally, $e_k\,\E$ ($0\leq k<\lambda$) are integral domains and we can compute a basis of~\eqref{Equ:CriticalBases} (by assumption a basis of~\eqref{Equ:AnnProblem} can be computed).
\end{proof}

\section{Solvers for $(R)\Pi\Sigma$-extensions}\label{Sec:GeneralSolvers}

We will now apply Theorem~\ref{Thm:IdempotentSolver} to a rather general class of difference rings built by basic
\rpisiE-ring extensions~\cite{DR1,DR3} that are defined over \pisiE-field extensions~\cite{Karr:81}. Before we can state  Theorem~\ref{Thm:QuotientRPSRingSolver} below, we will present more details on the underlying construction.

\begin{definition}\label{Def:APSExt}
	A difference ring $(\E,\sigma)$ is called an \emph{\rpisiE-ring extension} of a difference ring
	$\dfield{\A}{\sigma}$ if 
	$\A=\A_0\leq\A_1\leq\dots\leq\A_e=\E$
	is a tower of ring extensions with $\const{\E}{\sigma}=\const{\A}{\sigma}$ where for all $1\leq i\leq
	e$ one of the following holds:
	\begin{itemize}
		\item $\A_i=\A_{i-1}[t_i]$ is a ring extension subject to the relation $t_i^{\nu}=1$ for some $\nu>1$ where $\frac{\sigma(t_i)}{t_i}\in(\A_{i-1})^*$ is a primitive  $\nu$th root of unity ($t_i$ is called an \emph{\rE-monomial}, and $\nu$ is called the \emph{order of the \rE-monomial});
		\item $\A_i=\A_{i-1}[t_i,t_i^{-1}]$ is a Laurent polynomial ring extension with $\frac{\sigma(t_i)}{t_i}\in(\A_{i-1})^*$ ($t_i$ is called a \emph{\piE-monomial});
		\item $\A_i=\A_{i-1}[t_i]$ is a polynomial ring extension with $\sigma(t_i)-t_i\in\A_{i-1}$ ($t_i$ is called an \emph{\sigmaE-monomial}).
	\end{itemize}
	Depending on the occurrences of the \rpisiE-monomials such an extension is also called a \emph{\rE-/\piE-/\sigmaE-/$R\Pi$-/$R\Sigma$-/\pisiE-ring extension}.
\end{definition}

For convenience we use $\A\lr{t}$ for three different meanings: it is the ring $\A[t]$ subject to the relation $t^{\nu}=1$ if $t$ is an \rE-monomial of order $\nu$, it is the polynomial ring $\A[t]$ if $t$ is a \sigmaE-monomial, or it is the Laurent polynomial ring $\A[t,t^{-1}]$ if $t$ is a \piE-monomial.
We will restrict \rpisiE-ring extensions further to basic \rpisiE-ring extensions~\cite{DR3}.

\begin{definition}\label{defn:simpleNestedAExtension}
	Let $\dfield{\E}{\sigma}$ be a \rpisiE-ring extension of $\dfield{\A}{\sigma}$ with $\E=\A\langle t_{1}\rangle\dots\langle t_{e}\rangle$.
	We define the \emph{product group} by
	\begin{multline*}
	[\A^*]_{\A}^{\E} := \{ f\,t_1^{m_1}\dots t_e^{m_e}|\, f\in\A^* \text{ and } m_{i}\in\Z\\
	\text{ where $m_i=0$ if $t_i$ is an $R\Sigma$-monomial}\}.
\end{multline*}
	Then $\dfield{\E}{\sigma}$ is called a \emph{basic \rpisiE-ring extension} of $\dfield{\A}{\sigma}$ if for all \piE-monomials $t_i$ we have $\tfrac{\sigma(t_{i})}{t_{i}}\in [\A^*]_{\A}^{\A\langle t_{1}\rangle\dots\langle t_{i-1}\rangle}$ and for all \rE-mono\-mials $t_i$ we have  $\frac{\sigma(t_{i})}{t_{i}}\in\const{\A}{\sigma}^*$.
\end{definition}


In the following we seek for algorithms that solve PLDEs in a basic \rpisiE-ring extension  $\dfield{\E}{\sigma}$ of a difference field $\dfield{\F}{\sigma}$ with constant field $\K$. By Lemma~2.22 and Proposition~2.23 in~\cite{DR3} it turns out that one can collect several \rE-monomials to one specific \rE-monomial. Thus we assume from now on that $\dfield{\E}{\sigma}$ has the form
\begin{equation}\label{Equ:SingleRPS}
\E=\F[y]\lr{t_1}\dots\lr{t_e}
\end{equation}
where $y$ is an \rE-monomial of order $\lambda$ with $\alpha:=\frac{\sigma(y)}{y}\in\K^*$ and where the $t_i$ with $1\leq i\leq e$ are \pisiE-monomials with $\sigma(t_i)=\alpha_i\,t_i+\beta_i$ (note that either $\alpha_i=1$ or $\alpha_i\in[\F^*]_{\F}^{\E}$ with $\beta_i=0$). 
Now take 
$$\tilde{e}_s=\tilde{e}_s(y):=\tprod_{\substack{j=0\\j\neq \lambda-1-s}}^{\lambda-1}(y-\alpha^{j})$$
for $0\leq s<\lambda$. Since $\alpha$ is a $\lambda$th primitive root of unity, we have that $\tilde{e}_s(\alpha^{\lambda-1-s})\neq0$. Thus we can define 
\begin{equation}\label{Equ:SpecificEi}
e_s=e_s(y):=\frac{\tilde{e}_{s}(y)}{\tilde{e}_{s}(\alpha^{\lambda-1-s})}
\end{equation}
for $0\leq s<\lambda$ which fulfill precisely the properties enumerated in Definition~\ref{Def:IdemPotentDR}.
In particular, by~\cite[Thm.~4.3]{DR3} (compare also~\cite[Corollary~1.16]{Singer:97} and~\cite{Singer:08}) it follows that $\dfield{\E}{\sigma}$ is an idempotent difference ring of order $\lambda$. In particular, it is constant-stable provided that the ground field $\dfield{\F}{\sigma}$ is constant-stable.

\begin{definition}\label{Def:ConstantStable}
	A difference ring (resp.\ field) $\dfield{\A}{\sigma}$ is called \emph{constant-stable} if $\const{\A}{\sigma^k}=\const{\A}{\sigma}$ for all $k\in\N\setminus\{0\}$.
\end{definition}

\begin{theorem}[{\cite[Thm.~4.3]{DR3}}]\label{Thm:Interlacing}
	Let $\dfield{\E}{\sigma}$ be a basic \rpisiE-ring extension of a difference field $\dfield{\F}{\sigma}$ with~\eqref{Equ:SingleRPS} where $y$ is an \rE-monomial of order $\lambda$ with $\alpha=\frac{\sigma(y)}{y}$. Let $e_0,\dots,e_{\lambda-1}$ be the idempotent, pairwise orthogonal elements defined in~\eqref{Equ:SpecificEi} (that sum up to one). Then:
	\begin{enumerate}
		\item We get the direct sum~\eqref{Equ:decomposition} of the rings $e_s\,\E$ with the multiplicative identities $e_s$.
		\item We have that $e_s\,\E=e_s\,\tilde{\E}$ with the integral domain
		\begin{equation}\label{Equ:TildeE}
		\tilde{\E}:=\F\lr{t_1}\dots\lr{t_e}.
		\end{equation}
		\item For $0\leq s<\lambda$, $\dfield{e_s\,\tilde{\E}}{\sigma^\lambda}$ is a basic \pisiE-ring extension of $\dfield{e_s\,\F}{\sigma^\lambda}$. 
		\item $\sigma$ is a difference ring isomorphism between $\dfield{e_s\,\tilde{\E}}{\sigma^\lambda}$ and $\dfield{e_{s+1\mmod \lambda}\tilde{\E}}{\sigma^\lambda}$.
		\item Further, if $\dfield{\F}{\sigma}$ is constant-stable, $\const{e_s\,\E}{\sigma^\lambda}=e_s\,\const{\F}{\sigma}$.
	\end{enumerate}
\end{theorem}

In this particular setting, the used constructions in Section~\ref{Sec:IPDR} and in Theorem~\ref{Thm:Interlacing} can be made more precise as follows.
For $f\in\E$, the projection of the first component can be computed by 
$$\pi(f):=\ssumB{i=0}{\lambda-1} e_i(y)f\Big|_{y\to\alpha^{\lambda-1}}.$$
Furthermore, define for $n\in\N$ and $f\in\F$ the $\sigma$-factorial
$$\dfact{f}{n}{\sigma}=\sprodB{i=0}{n-1}\sigma^i(f).$$
Then with $\sigma(t_i)=\alpha\,t_i+\beta_i$ (recall that $\alpha_i=1$ or $\beta_i=0$) we get $\sigma^{\lambda}(t_i)=\tilde{\alpha}_i\,t_i+\tilde{\beta}_i$ with
$\tilde{\alpha}_i=\dfact{\alpha}{i}{\sigma}$ and $\tilde{\beta}_i=\sum_{l=0}^{i-1}\sigma^l(\beta)$.
In particular, we can define for $0\leq s<\lambda$ the ring automorphism $\sigma_s:\tilde{\E}\to\tilde{E}$ with $\sigma_s(f)=\sigma^{\lambda}(f)$ for $f\in\F$ and 
\begin{equation}\label{Equ:Sigmas}
\sigma_s(t_i)=\tilde{\alpha}_i\,t_i+(\tilde{\beta}_i|_{y\to\alpha^{\lambda-1-s}})
\end{equation}
for all $1\leq i\leq e$. Then $\dfield{\tilde{\E}}{\sigma_s}$ and $\dfield{e_s\tilde{\E}}{\sigma^{\lambda}}$ are isomorphic with the difference ring isomorphism $\tau:\tilde{\E}\to e_s\tilde{\E}$ with $\tau(f)=e_s\,f$ for $f\in\tilde{E}$; for further details we refer to~\cite[page~639]{DR3}. In the following we prefer to work with $\dfield{\tilde{\E}}{\sigma_s}$ instead of $\dfield{e_s\E}{\sigma^{\lambda}}$. Note that this representation is also more convenient for implementations.

As observed in Theorem~\ref{Thm:Interlacing} we obtain the \pisiE-ring extension $\dfield{\tilde{E}}{\sigma_s}$ of $\dfield{\F}{\sigma^{\lambda}}$ where $\tilde{\E}$ is an integral domain. Thus we can take the quotient field $Q(\tilde{E})=\F(t_1)\dots(t_e)$ and by naturally extending $\sigma_s:\tilde{\E}\to\tilde{\E}$ to $\sigma'_s:Q(\tilde{\E})\to Q(\tilde{\E})$ with $\sigma'_s(\frac{a}{b})=\frac{\sigma_s(a)}{\sigma_s(b)}$ we get a difference field $\dfield{Q(\tilde{\E})}{\sigma'_s}$; from now on we do not distinguish anymore between $\sigma'_s$ and $\sigma_s$. 

Finally, we take the ring of fractions 
$Q(\E)=\{\tfrac{a}{b}\mid a\in\E, b\in\E^*\}$
which can be written in terms of the idempotent representation
\begin{equation}\label{Equ:QuotRep}
Q(\E)=e_0\,Q(\tilde{\E})\oplus\dots\oplus e_{n-1}\,Q(\tilde{\E}).
\end{equation}
In particular, we can extend the automorphism $\sigma:\E\to\E$ to $\sigma:Q(\E)\to Q(\E)$ by mapping $f=e_0\,f_0+\dots+e_{\lambda-1}\,f_{\lambda-1}$ with $f_i\in\tilde{\E}$ to $\sigma(f)=e_0\,\sigma(f_{\lambda-1})+ e_1\sigma(f_0)+\dots+e_{\lambda-1}f_{\lambda-2}$; compare~\cite[Sec.~1.3 ]{Singer:99} and~\cite[Cor.~6.9]{Singer:08}.

Summarizing, also $\dfield{Q(\E)}{\sigma}$ is an idempotent difference ring of order $\lambda$ as introduced in Definition~\ref{Def:IdemPotentDR} and it seems naturally to apply Theorem~\ref{Thm:IdempotentSolver} to this more general situation.
Here we note (compare also~\cite[Prop.~66]{DR3}) that each component $\dfield{\tilde{\E}}{\sigma_s}$ for $0\leq s<\lambda$ is actually a special case of a  \pisiE-field-extension~\cite{Karr:81,Karr:85}.

\begin{definition}\label{Def:PSFieldExt}
	A difference field $(\F,\sigma)$ is called a \emph{\pisiE-field extension} of a difference field
	$\dfield{\HH}{\sigma}$ if 
	$\HH=\HH_0\leq\HH_1\leq\dots\leq\HH_e=\F$
	is a tower of field extensions with $\const{\F}{\sigma}=\const{\HH}{\sigma}$ where for all $1\leq i\leq
	e$ one of the following holds:
	\begin{itemize}
		\item $\HH_i=\HH_{i-1}(t_i)$ is a rational function field extension with $\frac{\sigma(t_i)}{t_i}\in(\HH_{i-1})^*$ ($t_i$ is called a \emph{\piE-field monomial});
		\item $\HH_i=\HH_{i-1}(t_i)$ is a rational function extension with $\sigma(t_i)-t_i\in\HH_{i-1}$ ($t_i$ is called a \emph{\sigmaE-field monomial}).
	\end{itemize}
\end{definition}

\noindent Here we will rely on the following property of \pisiE-field extensions; the first statement has been shown in~\cite{Karr:81} for \pisiE-fields. The second statement appears also in~\cite{SchneiderProd:20}.

\begin{proposition}\label{Prop:ConstantStable}
	Let $\dfield{\E}{\sigma}$ be a \pisiE-field/\pisiE-ring extension of a difference field $\dfield{\F}{\sigma}$ with $\K=\const{\F}{\sigma}$. Then:
	\begin{enumerate}
		\item For $k>1$, $\dfield{\E}{\sigma^{k}}$ is a \pisiE-field/\pisiE-ring extension of $\dfield{\F}{\sigma^{k}}$.
		\item If $\dfield{\F}{\sigma}$ is constant-stable, $\dfield{\E}{\sigma}$ is constant-stable.
	\end{enumerate}
\end{proposition}
\begin{proof}
	(1) Let $k>1$ and suppose that there is an $a\in\E\setminus\F$ with with $\sigma^k(a)=a$. Define $h=a\,\sigma(a)\dots\sigma^{k-1}(a)$. By 
	\cite[Lemma~31.(3)]{SchneiderProd:20} if follows that $h\not\in\F$. Since $\frac{\sigma(h)}{h}=\frac{\sigma^k(a)}{a}=1$, it follows that $h\in\const{\E}{\sigma}=\const{\F}{\sigma}\subseteq\F$, a contradiction. Note that for any \sigmaE-monomial $t$ with $\sigma(t)=t+\beta$ we have $\sigma^k(t)-t=\sum_{i=0}^{k-1}\sigma^i(\beta)$ and for any \piE-monomial $t$ with $\sigma(t)=\alpha\,t$ we have $\sigma^k(t)/t=\dfact{a}{k}{\sigma}$. Thus the automorphism $\sigma^k$ satisfies the requirements and consequently $\dfield{\E}{\sigma^k}$ is a \pisiE-field/\pisiE-ring extension of $\dfield{\F}{\sigma^k}$.\\
	(2) Suppose that $\dfield{\F}{\sigma}$ is constant-stable and let $k>1$. Then $\const{\F}{\sigma^k}=\const{\F}{\sigma}$. By statement (1), $\const{\F}{\sigma}=\const{\E}{\sigma}$ and thus $\const{\E}{\sigma^k}=\const{\E}{\sigma^k}$. Hence $\dfield{\E}{\sigma}$ is constant-stable.
\end{proof}

In this particular scenario, we can refine Theorem~\ref{Thm:IdempotentSolver} as follows.

\begin{proposition}\label{Prop:BasicStrategyInBasicExt}
Let $\dfield{\F}{\sigma}$ be a constant-stable difference field with constant field $\K$, and let $\dfield{\E}{\sigma}$ with~\eqref{Equ:SingleRPS} be a basic \rpisiE-ring extension with only one \rE-monomial $y$ with $\frac{\sigma(y)}{y}\in\K$ of order $\lambda$. Then one can solve non-degenerated PLDEs in  $\dfield{\E}{\sigma}$ (resp.\ in $\dfield{Q(\E)}{\sigma}$) if $\dfield{\E}{\sigma}$ is computable and one can solve PLDEs in the \pisiE-ring extension $\dfield{\tilde{\E}}{\sigma_0}$ (resp.\ \pisiE-field extension $\dfield{\Q(\tilde{\E})}{\sigma_0}$) of $\dfield{\F}{\sigma^{\lambda}}$ with~\eqref{Equ:TildeE}.
\end{proposition}
\begin{proof}
$\dfield{\tilde{\E}}{\sigma_0}$ is a basic \pisiE-ring extension of $\dfield{\F}{\sigma^{\lambda}}$ by Theorem~\ref{Thm:Interlacing}.(3), and thus taking the quotient field $Q(\tilde{\E})$, $\dfield{\Q(\tilde{\E})}{\sigma_0}$ is a \pisiE-field extension of $\dfield{\F}{\sigma^{\lambda}}$ by iterative application of~\cite[Cor.~2.6]{DR3}. 
Since $\dfield{\F}{\sigma}$ is constant-stable, we get $\const{Q(\tilde{\E})}{\sigma_0}=\const{\tilde{\E}}{\sigma_0}=\const{\F}{\sigma^{\lambda}}=\const{\F}{\sigma}=\K$.
Finally, since we can solve PLDEs in $\dfield{\tilde{\E}}{\sigma_0}$ (resp.\ in $\dfield{Q(\tilde{\E})}{\sigma_0}$) by assumption, we can apply Theorem~\ref{Thm:IdempotentSolver} and can compute all solutions of non-degenerated PLDEs in $\dfield{\E}{\sigma}$ (resp.\ in $\dfield{Q(\E)}{\sigma}$).
\end{proof}

\subsection{The general case: basic \rpisiE-ring extensions over \pisiE-field extensions}

To activate Proposition~\ref{Prop:BasicStrategyInBasicExt} we have to take an appropriate difference field  $\dfield{\F}{\sigma}$ such that (1) it is constant-stable and such that (2) PLDEs can be solved in $\dfield{\tilde{\E}}{\sigma_0}$. As it turns out, both properties can be fulfilled if $\dfield{\F}{\sigma}$ itself is a \pisiE-field extension of a difference field $\dfield{\GG}{\sigma}$ that enjoys certain algorithmic properties. In this situation, the first property can be settled using Proposition~\ref{Prop:ConstantStable} from above.
To deal with the second property, we will introduce the following problems; a certain subset of them have been introduced originally in~\cite{Schneider:06d} (by analyzing Karr's (telescoping) algorithms in~\cite{Karr:81}).

\begin{definition}[\cite{ABPS:21}]\label{Def:computable}
	A difference field $\dfield{\F}{\sigma}$ with constant field $\K$ is
	\emph{$\sigma$-com\-pu\-ta\-ble} if $\dfield{\E}{\sigma}$ is computable and the following holds.
	
	\begin{enumerate}
		\item One can factor multivariate polynomials
		over $\F$.
		
		\item $\dfield{\F}{\sigma^s}$ is \emph{torsion free} for any $s\in\Z^*$, i.e.,

\vspace*{-0.45cm}

		\begin{equation*}
		\quad\quad\forall s,r\in\Z^*\;\forall f,g\in \F^*:\,
		f=\tfrac{\sigma^s(g)}{g} \wedge f^r=1\Rightarrow f=1.
		\end{equation*}
		
\vspace*{-0.1cm}

		\item The \emph{$\Pi$-Regularity problem} is solvable:
		Given $\dfield{\F}{\sigma}$ and $f,g\in\F^*$; find, if possible, an $n\geq0$ with
		$\dfact{f}{n}{\sigma}=g$.
		
		\item The \emph{$\Sigma$-Regularity problem} is solvable: Given $\dfield{\F}{\sigma}$, $r\in\Z^*$, $f,g\in \F^*$;
		find, if possible, $n\geq0$ with
		$\dfact{f}{0}{\sigma^r}+\cdots+\dfact{f}{n}{\sigma^r}=g$.
		
		\item The \emph{parameterized pseudo-orbit problem} is solvable:
		Given $\vect{f}=(f_1,\dots,f_n)\in (\F^*)^d$;
		compute a $\Z$-basis of the module

\vspace*{-0.45cm}

		$$\quad\quad M(\vect{f},\F)=\{(z_1,\dots,z_d)\in\Z^n\mid \exists g\in\F^* \frac{\sigma(g)}{g}=f_1^{z_1}\dots f_d^{z^d}\}.$$

\vspace*{-0.15cm}

		\item There is an algorithm that can \emph{compute all the hypergeometric candidates for
			equations with coefficients in $\dfield{\F}{\sigma}$}: Given a nonzero operator $L \in\F[\sigma]$;
		compute a finite set $S \subset\F$ such that for any $r \in\F^*$,
		if $\sigma - r$ is a right factor of $L$ in $\F[\sigma]$,
		then $r = u \frac{\sigma(v)}{v}$ for some $u \in S$ and $v \in\F^*$.
		\item \emph{PLDEs are solvable} in $\dfield{\F}{\sigma}$: Given $\vect{0}\neq\vect{a}\in\F^{m+1}$, $\vect{f}\in\F^d$; compute a $\K$-basis of $V(\vect{a},\vect{f},\F)$.
	\end{enumerate}
\end{definition}

Then using the brandnew framework summarized in~\cite[Thm.~10]{ABPS:21}, we obtain the following result which has been implemented within the summation package~\texttt{Sigma}.

\begin{theorem}[\cite{ABPS:21}]\label{Thm:PiSiSolver}
	Let $\dfield{\E}{\sigma}$ be a (nested) \pisiE-field extension of $\dfield{\F}{\sigma}$. If $\dfield{\F}{\sigma}$ is $\sigma$-computable, then also $\dfield{\E}{\sigma}$ is $\sigma$-computable.
\end{theorem}

In particular, using~\cite{Schneider:06d,ABPS:21} (based on~\cite{Karr:81}) the properties given in Definition~\ref{Def:computable} simplify in the special case $\sigma=\text{id}$ as follows.

\begin{theorem}\label{Thm:GroundField}
	Let $\K$ be a computable field where
	\begin{enumerate}
		\item polynomials can be factored in $\K[t_1,\dots,t_e]$,
		\item a basis of $\{(z_1,\dots,z_d)\in\Z^d\mid 1=\prod_{i=1}^d c_i^{z_i}\}$ can be computed,
		\item one can recognize if $c \in k$ is an integer,
	\end{enumerate}
then $\dfield{\K}{\sigma}$ with $\const{K}{\sigma}=\K$  is $\sigma$-computable.
\end{theorem}

%
%
%
%
%
%
%

We can now state our first algorithmic framework to solve non-degenerated PLDEs in $(R)\Pi\Sigma$-extensions..

\begin{theorem}\label{Thm:QuotientRPSRingSolver}
	Let $\dfield{\F}{\sigma}$ be a \pisiE-field extension of a difference field $\dfield{\GG}{\sigma}$ and let
	$\dfield{\E}{\sigma}$ be a basic \rpisiE-ring extension of $\dfield{\F}{\sigma}$ 
	with one \rE-monomial $y$ with $\frac{\sigma(y)}{y}\in\const{\F}{\sigma}$ of order $\lambda$. Then one can solve non-degenerated PLDEs in the quotient ring $\dfield{Q(\E)}{\sigma}$ or in $\dfield{\E}{\sigma}$ if one of the following holds:
	\begin{enumerate} 	
		\item $\dfield{\GG}{\sigma}$ is constant-stable and
		$\dfield{\GG}{\sigma^{\lambda}}$ is $\sigma$-computable.
		\item $\const{\GG}{\sigma}=\GG$ satisfies the properties in Theorem~\ref{Thm:GroundField}.  
		\item $\const{\GG}{\sigma}=\GG$ is a rat.\ function field over an alg.\ number field. 
	\end{enumerate}  
\end{theorem}
\begin{proof}
	(1) Since $\dfield{\GG}{\sigma}$ is constant-stable, it follows that $\dfield{\F}{\sigma}$ is constant-stable by Proposition~\ref{Prop:ConstantStable}.(2). Furthermore, $\dfield{\F}{\sigma^{\lambda}}$ is a \pisiE-field extension of $\dfield{\GG}{\sigma^{\lambda}}$ by Proposition~\ref{Prop:ConstantStable}.(1) and thus
	$\dfield{Q(\tilde{\E})}{\sigma_0}$ is a \pisiE-field extension of $\dfield{\GG}{\sigma^{\lambda}}$. Since $\dfield{\GG}{\sigma^{\lambda}}$ is $\sigma$-computable, we conclude with Theorem~\ref{Thm:PiSiSolver} 
	that also $\dfield{Q(\tilde{\E})}{\sigma_0}$ is $\sigma$-computable, in particular property~(7) in Definition~\ref{Def:computable} holds. Hence we can apply Proposition~\ref{Prop:BasicStrategyInBasicExt} and can solve all non-degenerated PLDEs in $\dfield{Q(\E)}{\sigma}$. Given a basis in $Q(\E)$ one can filter out a basis of the subspace in $\E$ by linear algebra\footnote{In Section~\ref{Sec:ImprovedVersions} we will provide improved algorithms to accomplish this task directly.}.\\
	(2) Since $\GG=\const{\GG}{\sigma}$, $\sigma|_{\GG}=\text{id}$. Thus $\dfield{\GG}{\sigma}$ is trivially constant-stable. In addition, if the properties of Theorem~\ref{Thm:GroundField} are fulfilled, $\dfield{\GG}{\sigma}$ is $\sigma$-computable and thus we can apply part (1).\\
	(3) By~\cite{Ge:93} and~\cite[Thm.~3.5]{Schneider:05c} if follows that the algorithms required in Theorem~\ref{Thm:GroundField} are available. Thus we can apply part (2).
\end{proof}

\begin{remark}\label{Remark:MixedCase}
	 Theorem~\ref{Thm:QuotientRPSRingSolver} (Case~3) covers, e.g., the \emph{rational} ($v=0$) or the \emph{mixed multibasic difference field} $\dfield{\GG}{\sigma}$ with $\GG=\K(x,x_1,\dots,x_v)$ where $\K=K(q_1\dots,q_v)$ is a rational function field ($K$ itself is a rational function field over an algebraic number field) and with $\sigma|_{\K}=\text{id}$, $\sigma(x)=x+1$ and $\sigma(x_i)=q_i\,x_i$ for $1\leq i\leq v$.
\end{remark}

\subsection{Simplified algorithms for special ring cases}\label{Sec:ImprovedVersions}


The PLDE solver summarized in Theorem~\ref{Thm:QuotientRPSRingSolver} assumes that $\dfield{\GG}{\sigma}$ is $\sigma$-computable. In the following we restrict ourselves to some interesting  sub-classes of \rpisiE-ring extensions where the $\Sigma$- and $\Pi$-regularity problem in Definition~\ref{Def:computable} (but also the hidden shift-equivalence problem within the tower of extensions) can be avoided.
As a consequence one ends up at lighter implementations where most of the highly recursive algorithms from~\cite{Karr:81} can be skipped.

Let $\dfield{\A\lr{t}}{\sigma}$ be a \pisiE-ring extension of $\dfield{\A}{\sigma}$ with constant field $\K=\const{\A}{\sigma}$. Assume in addition that $\A$ is an integral domain and that one can solve PLDEs in $\dfield{\A}{\sigma}$. Then we can apply the following tactic~\cite{Schneider:05a} (which is inspired by~\cite{Karr:81} and is also the backbone strategy in~\cite{ABPS:21}) to find a basis of $V=V(\vect{a},\vect{f},\A\lr{t})$ with $\vect{0}\neq\vect{a}=(a_0,\dots,a_m)\in\A\lr{t}^{m+1}$ and $\vect{f}=(f_1,\dots,f_d)\in\A\lr{t}^d$. First, we bound the degree of the possible solutions: namely, we compute $a,b\in\Z$ such that
for any $(c_1,\dots,c_d,\sum_{k=a'}^{b'} g_i\,t^i)\in V$ we have $a\leq a'$ and $b'\leq b$; if $t$ is a \sigmaE-monomial we set $a=0$ and search for $b$ only. Then given such bounds $a,b$, we make the ansatz~\eqref{Equ:PLDE} with unknown $c_0,\dots,c_d\in\K$ and $g=\sum_{k=a'}^{b'} g_i\,t^i$ with unknown $g_a,\dots,g_b\in\A$. By comparing coefficients in~\eqref{Equ:PLDE} w.r.t.\ to the highest arising term we obtain a PLDE in $\dfield{\A}{\sigma}$ which has $c_1,\dots,c_d$ and $g_b\in\A$ as solution. Solving this PLDE yields all possible candidates for $g_b$. Thus plugging these choices into~\eqref{Equ:PLDE} we can proceed recursively (by degree reduction) to nail down $g_b$ and the remaining coefficients $g_a,\dots,g_{b-1}$.

Due to~\cite[Theorem~7]{ABPS:21} it follows that one can determine $b\in\N$ and $a=0$ for a \sigmaE-monomial $t$ if one can solve PLDEs in $\dfield{\A}{\sigma}$. Thus activating this machinery recursively yields the following result.

\begin{proposition}\label{Prop:LiftSigmaSolver}
	If one can solve PLDEs in $\dfield{\A}{\sigma}$, then one can solve PLDEs in a \sigmaE-extension $\dfield{\A\lr{t_1}\dots\lr{t_e}}{\sigma}$ of $\dfield{\A}{\sigma}$
\end{proposition}

For \piE-monomials one can utilize~\cite[Theorem~6]{ABPS:21} to compute the above bounds $a,b\in\Z$. If one applies this machinery recursively (as for \sigmaE-monomials) one ends up at the requirement that the ground ring is $\sigma$-computable. In a nutshell, we rediscover the ring version of Theorem~\ref{Thm:QuotientRPSRingSolver} -- but this time we solve it directly without computing first all solutions in its quotient field. 

In the following we adapt slightly the proof steps of~\cite[Theorem~6]{ABPS:21} yielding the more flexible Lemma~\ref{Lemma:PiBounds}. For its proof, we need in addition the following result.

\begin{lemma}\label{Lemma:atMostOneSol}
	Let $\dfield{\F\lr{t_1}\dots\lr{t_e}}{\sigma}$ be a \piE-ring extension of $\dfield{\F}{\sigma}$ with $\alpha_i=\frac{\sigma(t_i)}{t_i}\in\F^*$. Let $V=M(\alpha_1,\dots,\alpha_e,u)$ for some $u\in\F^*$.
	Then $V=\emptyset$ or $V=\Z(\lambda_1,\dots,\lambda_{e+1})$ for some $\lambda_i\in\Z$ with $\lambda_{e+1}>0$.
\end{lemma}
\begin{proof}
	Suppose that $V\neq\emptyset$. Suppose further that we can take $\vect{0}\neq(\lambda_1,\dots,\lambda_e,0)\in V$. Then we get $g\in\F^*$ with $\frac{\sigma(g)}{g}=\alpha_1^{\lambda_1}\dots\alpha_e^{\lambda_e}$, not all $\lambda_i$ being zero, which is not possible by~\cite[Thm.~9.1]{Schneider:10c}. Consequently, for any nonzero vector in $V$ we conclude that the last entry must be nonzero.
	Now take 
	$\vect{\lambda}=(\lambda_1,\dots,\lambda_{e+1}),\vect{\mu}=(\mu_1,\dots,\mu_{e+1})\in V\setminus\{\vect{0}\}$. Then $\lambda_{e+1},\mu_{e+1}\neq0$. In particular, $\vect{a}=\mu_{e+1}\vect{\lambda}-\lambda_{e+1}\vect{\mu}\in V$. Since the last entry of $\vect{a}$ is zero, it follows that $\vect{a}=\vect{0}$. Hence two nonzero vectors are linearly dependent and it follows that $V=\{(\lambda_1,\dots,\lambda_{e+1})\}\Z$ with $\lambda_{e+1}\neq0$. If $\lambda_{e+1}<0$, we can choose the alternative generator $(-\lambda_1,\dots,-\lambda_{e+1})$ with $-\lambda_{e+1}>0$.
\end{proof}

\begin{lemma}\label{Lemma:PiBounds}
	Let $\dfield{\E}{\sigma}$ with $\E=\F\lr{t_1}\dots\lr{t_e}$ be a \piE-ring extension of $\dfield{\F}{\sigma}$ with $\alpha_i=\frac{\sigma(t_i)}{t_i}\in\F^*$.
	If one can solve the parameterized pseudo problem in $\dfield{\F}{\sigma}$ and can find all hypergeometric candidates in $\dfield{\F}{\sigma}$, one can bound the degrees of the solutions w.r.t.\ $t_e$.
\end{lemma}
\begin{proof}
	Let $\vect{f}=(f_1,\dots,f_d)\in\E^d$ and $(a_0,\dots,a_m)\in\E^{m+1}$ with $a_0\,a_m\in\E^*$ and suppose that $g\in\E$ is a solution of~\eqref{Equ:PLDE}. Let $l_e$ be the highest degree in $g$ w.r.t.\ $t_e$.
	In the following we take the lexicographic order $<$ on $M=\{t_1^{n_1}\dots t_e^{n_e}\mid n_1,\dots,n_e\in\Z\}$ with $t_1<t_2<\dots <t_e$, and $t_i^a<t_i^b$ iff $a<b$.
	Let $\tilde{g}=h\,t_1^{\lambda_1}\dots t_e^{\lambda_e}$ be the highest term in $g$; note that $\lambda_e=l_e$.
	Further, let $\mu=t_1^{\mu_1}\dots t_e^{\mu_e}\in M$ be the largest monomial of the coefficients in $\vect{a}$, and let $\tilde{a}_i\in\F$ for $0\leq i\leq m$ be the corresponding coefficient of $\mu$; note that one of the $\tilde{a}_i$ is nonzero. Take $L:=\tilde{a_0}+\tilde{a_1}\,\sigma+\dots+\tilde{a_m}\,\sigma^m\in\F[\sigma]$.\\
	Now suppose that $L(\tilde{g})=0$ and define $\alpha=\frac{\sigma(h)}{h}\,\alpha_1^{\lambda_1}\dots\alpha_e^{\lambda_e}\in\F^*$. Note that for $\tilde{L}=\sigma-\alpha\in\F[\sigma]$ we have $\tilde{L}(\tilde{g})=0$ by construction.
	Now we follow the arguments of~\cite[Lemma~2]{ABPS:21}. 
	Let $L=Q\,\tilde{L}+R$ be the right-division of $L$ by $\tilde{L}$ with $Q\in\F[\sigma]$ and $R\in\F$.
	Since $0=L(\tilde{g})=Q\,\tilde{L}(\tilde{g})+R=R$, $\tilde{L}$ is a right-factor of $L$.
	By assumption we can compute a set $S$ which contains all hypergeometric candidates of $\tilde{L}$. Thus we can take $u\in S$ with
	$\frac{\sigma(h)}{h}\alpha_1^{\lambda_1}\dots\alpha_e^{\lambda_e}=\alpha=u\frac{\sigma(w)}{w}.$
	Consequently, we get
	$\alpha_1^{\lambda_1}\dots\alpha_e^{\lambda_e}u^{-1}=\frac{\sigma(w')}{w'}$
	for some $w'\in\F^*$. Now compute a basis $B_u$ of $V_u=M(\alpha_1,\dots,\alpha_e,u^{-1};\F)$. 
	By Lemma~\ref{Lemma:atMostOneSol} we can assume that $B_u=\{(\nu_{u,1},\dots,\nu_{u,e+1})\}\in\Z^{e+1}$ with $\nu_{u,e+1}>0$. Note that it follows even that $\nu_{u,e+1}=1$ and $l_e=\lambda_e=\nu_{u,e}$.\\ 
	Thus to bound the leading coefficient w.r.t.\ $t_e$ we proceed as follows: We loop trough all $u\in S$ and compute a basis $B_u$ of $V_u$ and take
	$$C=\max\{\nu_{u,e}\mid V_u=(\nu_{u,1},\dots,\nu_{u,e},1)\Z\text{ for }u\in S\}.$$
	Summarizing, let $l_e$ be the highest degree in the solution $g$ w.r.t.\ $t_e$ and let $\tilde{g}=ht_1^{\lambda_1}\dots t_e^{\lambda_e}$ be the highest term in $g$. If $\tilde{L}(\tilde{g})=0$. then $l_e\in C$, i.e., $C\neq\emptyset$ and $l_e\leq \max(C)$. Otherwise, if $C=\emptyset$ or $\tilde{L}(\tilde{y})\neq0$, we conclude as follows. We note that $\tilde{L}(\tilde{y})=h't_1^{\lambda_1}\dots t_e^{\lambda_e}$ for some $h'\in\F^*$. Since $\tilde{y}$ is the largest term in our solution $y$ and since $\tilde{L}$ is the contribution of the highest term in~\eqref{Equ:PLDE}, it follows by coefficient comparison in~\eqref{Equ:PLDE} that $\tilde{L}(\tilde{y})t_1^{m_1}\dots t_e^{m_e}=h't_1^{m_1+\lambda_1}\dots t_e^{m_e+\lambda_e}$ for some $h'\in\F^*$ must arise in $c_1\,f_1+\dots +c_d\,f_d$. Thus, if $d_i$ is the largest exponent in $f_i$ w.r.t.\ $t_e$, we get $m_e+\lambda_e<\max(d_1,\dots,d_e)$. In conclusion, if $C=\emptyset$, we get $l_e=\lambda_e\leq\max(d_1,\dots,d_e)-m_e=:b$. Otherwise, we conclude that $l_e=\lambda_e\leq\max(\max(C),b)$.
	Similarly, we can bound the lowest term in $g$ by repeating this procedure and taking the order $<$ with $t_1<t_2<\dots <t_e$ and $t_i^a<t_i^b$, iff $a>b$ and replacing the $\max$ operation with the $\min$ operation, etc.
\end{proof}

\noindent The above results yield, in comparison to Theorem~\ref{Thm:QuotientRPSRingSolver}, the following less general but simpler (less recursive algorithms) and more flexible (less requirements) toolbox to solve PLDEs in \pisiE-ring extensions.

\begin{theorem}\label{Thm:PiSigmaRingSolverRefined}
	Let $\dfield{\E}{\sigma}$ with $\E=\F\lr{t_1}\dots\lr{t_e}$ be a \pisiE-ring extension of a difference field $\dfield{\F}{\sigma}$ where for all \piE-monomials $t_i$ we have  $\frac{\sigma(t_i)}{t_i}\in\F^*$. If one can solve PLDEs, the parameterized pseudo-orbit problem and hypergeometric candidates in $\dfield{\F}{\sigma}$, then one can solve PLDEs  in $\dfield{\E}{\sigma}$.
\end{theorem}
\begin{proof}
	By reordering we may assume that $\A=\F\lr{t_1}\dots\lr{t_l}$ contains precisely the \piE-monomials of $\E$ and that the $t_{l+1},\dots,t_e$ form all \sigmaE-monomials. By Lemma~\ref{Lemma:PiBounds} we can bound the degree of the solutions w.r.t.\ $t_l$. By iteration (recursion) we can thus solve PLDEs in $\dfield{\A}{\sigma}$. Finally, with Prop.~\ref{Prop:LiftSigmaSolver} we can solve PLDEs in $\dfield{\E}{\sigma}$.
\end{proof}

\noindent Combining Theorem~\ref{Thm:PiSigmaRingSolverRefined} with Proposition~\ref{Prop:BasicStrategyInBasicExt} yields Theorem~\ref{Thm:SpecialRPiSiSolver}.

\begin{theorem}\label{Thm:SpecialRPiSiSolver}
	Let $\dfield{\E}{\sigma}$ be an \rpisiE-ring extension of a constant-stable difference field $\dfield{\F}{\sigma}$ 
	with one \rE-monomial $y$ with $\frac{\sigma(y)}{y}\in\const{\F}{\sigma}$ of order $\lambda$ and where for each \piE-monomial $t$ in the extension $\E$ of $\F$ we have $\frac{\sigma(t)}{t}\in\F^*$. If one can solve PLDEs, solve the parameterized pseudo-orbit problem and can find all hypergeometric candidates in $\dfield{\F}{\sigma^{\lambda}}$, one can solve non-degenerated PLDEs in $\dfield{\E}{\sigma}$.
\end{theorem}

\noindent Using results of~\cite{Bauer:99}, this PLDE solver is, e.g., applicable if one specializes $\F$ to the mixed multibasic case introduced in Remark~\ref{Remark:MixedCase}.

\section{Example}\label{Sec:Example}

We will illustrate the whole machinery by solving the recurrence: 
\small
\begin{align*}
&\Big[(1+n) (2+n) \big(
        \big(
                2
                +n
                +(1+n) 
                \ssumB{i=1}n \frac{1}{i}
        \big) (-1)^n
        -(1+n)^2 
        \ssumB{i=1}n \frac{(-1)^{i}}{i}
\big)\Big]\,G(n)\\[-0.1cm]
&+\Big[(1+n) (2+n)  \big(
        \big(
                2
                +n
                +2 (1+n) 
                \ssumB{i=1}n \frac{1}{i}
        \big) (-1)^n
        -(1+n) 
        \ssumB{i=1}n \frac{(-1)^{i}}{i}
\big)\Big]\,G(n+1)\\[-0.1cm]
&+\Big[(1+n)^2 (2+n) \big(
        (-1)^n 
        \ssumB{i=1}n \frac{1}{i}
        +n 
        \ssumB{i=1}n \frac{(-1)^{i}}{i}
\big)\Big]G(n+2)\\[-0.1cm]
&=
(2+n)^2
+(1+n) 
\ssumB{i=1}n \frac{1}{i}
-2 (1+n)^3 (-1)^n 
\ssumB{i=1}n \frac{(-1)^{i}}{i}.
\end{align*}
\normalsize
Internally, we represent the recurrence in the basic \rpisiE-ring extension $\dfield{\E}{\sigma}$ of $\dfield{\Q(x)}{\sigma}$ with $\E=\Q(x)[y][s][\bar{s}]$ where $\sigma(x)=x+1$, $\sigma(y)=-y$, $\sigma(s)=s+\frac{1}{x+1}$ and $\sigma(\bar{s})=\bar{s}+\frac{-y}{x+1}$. Note that $\dfield{\E}{\sigma}$ is an idempotent difference ring of order $2$ with $e_0=\frac{1-y}2$ and $e_1=\frac{1+y}2$. Then the recurrence turns into $\sum_{i=0}^2 a_i \sigma^i(g)=\varphi$ with
\small
\begin{align*}
 \vect{a}=&\Big((1+x) (2+x) (-\bar{s} (1+x)^2+(2+s+x+s x) y), (1+x) (2+x) \\[-0.2cm]
   &(-\bar{s} (1+x)+(2+x+2 s (1+x)) y),  (1+x)^2 (2+x) (\bar{s} x+s y)\Big),\\
\varphi=&s (1+x)+(2+x)^2-2 \bar{s} (1+x)^3 y.
\end{align*}
\normalsize
With Theorem~\ref{Subops} we compute with the package \texttt{HarmonicSums}~\cite{HarmonicSums} for the first component the equation $\sum_{i=0}^2b_{1,i}\sigma^{2i}(g_0)=\varphi_0$ with $\vect{b_0}=(b_{i,0},b_{i,1})$ where
\small
\begin{align*}
 \vect{b_0}=&\Big( x (29+33 x+11 x^2+x^3+2 s (6+11 x+6 x^2+x^3)+\bar{s} (6+11 x+6 x^2\\[-0.2cm]
    &+x^3t)),-x (41+49 x+18 x^2+2 x^3+4 s (6+11 x+6 x^2+x^3)+2 \bar{s} (6\\[-0.2cm]
    &+11 x+6 x^2+x^3)),x (2+x) (3+x)\biggr),\\[-0.3cm]     
\varphi_0=&\frac{-x}{(1+x) (2+x) (4+x)} \biggl(292+559 x+387 x^2+114 x^3+12 x^4+4 s(22\\[-0.2cm]
          &+53 x+45 x^2+16 x^3+2 x^4)+2 \bar{s} (22+53 x+45 x^2+16 x^3+2 x^4)\Big).
\end{align*}
\normalsize
A similar linear difference equation can be computed for the second component. Solving these equations (activating, e.g., Theorem~\ref{Thm:PiSigmaRingSolverRefined} with \texttt{Sigma}~\cite{Schneider:07a}) leads to the solutions
\begin{align*}
t_0&=s+c_1+c_2 (s+\bar{s}+2 x-4 s x-2 \bar{s} x),\\
t_1&=-s+d_1+d_2 (-s+\bar{s}-2 x+4 s x-2 \bar{s} x),
\end{align*}
for $c_1,c_2,d_1,d_2\in \Q$. Plugging $g:=e_0\,t_0+e_1\,t_1$ into $\sum_{i=0}^2 a_i \sigma^i(g)=\varphi$ gives us constraints for the constants (compare Theorem~\ref{Thm:IdempotentSolver}) and we find
$d_1=-c_1$ and $d_2=c_2$. These solutions can be combined to the general solution
$$
-s y-c_1 y +c_2 (\bar{s}-2 \bar{s} x-s y-2 x y-4 s x y),
$$
of $\sum_{i=0}^2 a_i \sigma^i(g)=\varphi$, i.e., $\{(0,y),(0,\bar{s}-2 \bar{s} x-s y-2 x y-4 s x y),(1,-s y)\}$ is a basis of
$V(\vect{a},(\varphi),\E)$. Finally, by reinterpreting the result in terms of sums and products we find the following general solution of the original recurrence:
\begin{align*}
&-\ssumB{i=1}n \frac{1}{i}(-1)^n-c_1(-1)^n\\[-0.2cm]
&\quad+c_2\Big(-2(-1)^nn-(1+4n)(-1)^n\ssumB{i=1}n \frac{1}{i}+\ssumB{i=1}n\frac{(-1)^i}{i}(1-2n)\Big).
\end{align*}

\section{Conclusion}\label{Sec:Conclusion}

We have considered idempotent difference rings (heavily used in the Galois theory of difference equations~\cite{Singer:97,Singer:08}) and derived a general toolbox to solve PLDEs in this setting. More precisely,
we introduced the notion of non-degenerated linear difference operators and showed that finding solutions for a given PLDE in difference rings with zero-divisors can be reduced to finding solutions in difference rings that are integral domains (see Theorems~\ref{Subops} and~\ref{Thm:IdempotentSolver}). In the second part of this article we provided two general PLDE solvers: Theorem~\ref{Thm:QuotientRPSRingSolver} for the most general case which assumes that rather strong properties hold in the ground field and Theorem~\ref{Thm:SpecialRPiSiSolver} which is less general, but where some of the complicated algorithmic assumptions can be dropped. In both cases, the inner core (Theorem~\ref{Thm:PiSiSolver}) is a PLDE solver for \pisiE-field extensions that has been elaborated in~\cite{ABPS:21} and implemented within \texttt{Sigma}.

Our notion of non-degenerated operators is motivated by our method to decompose the desired solution. An interesting question is if there are equivalent (or even more flexible definitions) that are easier to verify. We also indicated that the decomposition method (implemented in the package \texttt{HarmonicSums}) works partially if the operator is degenerated. Further investigations in this direction, also connected to the dimension of the solution space, would be highly interesting. Finally, we are strongly motivated to generalize our PLDE solver summarized in  Theorem~\ref{Thm:SpecialRPiSiSolver} further to more general classes of (basic) \rpisiE-ring extensions.


\begin{thebibliography}{10}
	
	\bibitem{HarmonicSums}
	J.~Ablinger.
	\newblock The package {H}armonic{S}ums: {C}omputer algebra and analytic aspects
	of nested sums.
	\newblock In {\em {Loops and Legs in QFT 2014}}, pages 1--10, 2014.
	
	\bibitem{Abramov:89a}
	S.~A. Abramov.
	\newblock Rational solutions of linear differential and difference equations
	with polynomial coefficients.
	\newblock {\em U.S.S.R. Comput. Math. Math. Phys.}, 29:7--12, 1989.
	
	\bibitem{ABPS:21}
	S.~A. Abramov, M.~Bronstein, M.~Petkov{\v s}ek, and C.~Schneider.
	\newblock On rational and hypergeometric solutions of linear ordinary
	difference equations in {$\Pi\Sigma^*$}-field extensions.
	\newblock {\em J. Symb. Comput., in press}, 2021.
	\newblock arXiv:2005.04944.
	
	\bibitem{Bauer:99}
	A.~Bauer and M.~Petkov{\v{s}}ek.
	\newblock Multibasic and mixed hypergeometric {Gosper}-type algorithms.
	\newblock {\em J.~Symbolic Comput.}, 28(4--5):711--736, 1999.
	
	\bibitem{BRS:16}
	J.~Bl\"umlein, M.~Round, and C.~Schneider.
	\newblock Refined holonomic summation algorithms in particle physics.
	\newblock In {\em {Advances in Computer Algebra. WWCA 2016.}}, volume 226 of
	{\em Springer Proceedings in Mathematics \& Statistics}, pages 51--91. 2018.
	
	\bibitem{Chyzak:00}
	F.~Chyzak.
	\newblock An extension of {Z}eilberger's fast algorithm to general holonomic
	functions.
	\newblock {\em Discrete Math.}, 217:115--134, 2000.
	
	\bibitem{Cohn:65}
	R.~M. Cohn.
	\newblock {\em Difference Algebra}.
	\newblock John Wiley \& Sons, 1965.
	
	\bibitem{Ge:93}
	G.~Ge.
	\newblock {\em Algorithms related to the multiplicative representation of
		algebraic numbers}.
	\newblock PhD thesis, Univeristy of California at Berkeley, 1993.
	
	\bibitem{Gosper:78}
	R.~W. Gosper.
	\newblock Decision procedures for indefinite hypergeometric summation.
	\newblock {\em Proc. Nat. Acad. Sci. U.S. A.}, 75:40--42, 1978.
	
	\bibitem{Singer:08}
	C.~Hardouin and M.~F. Singer.
	\newblock Differential {G}alois theory of linear difference equations.
	\newblock {\em Math. Ann.}, 342(2):333--377, 2008.
	
	\bibitem{Singer:99}
	P.~A. Hendriks and M.~F. Singer.
	\newblock Solving difference equations in finite terms.
	\newblock {\em J.~Symbolic Comput.}, 27(3):239--259, 1999.
	
	\bibitem{Karr:81}
	M.~Karr.
	\newblock Summation in finite terms.
	\newblock {\em J.~ACM}, 28:305--350, 1981.
	
	\bibitem{Karr:85}
	M.~Karr.
	\newblock Theory of summation in finite terms.
	\newblock {\em J.~Symb. Comput.}, 1:303--315, 1985.
	
	\bibitem{Schneider:06d}
	M.~Kauers and C.~Schneider.
	\newblock Indefinite summation with unspecified summands.
	\newblock {\em Discrete Math.}, 306(17):2021--2140, 2006.
	
	\bibitem{Mallinger}
	C.~Mallinger.
	\newblock Algorithmic manipulations and transformations of univariate holonomic
	functions and sequences.
	\newblock Master's thesis, RISC, JKU Linz, 1996.
	
	\bibitem{Petkov:92}
	M.~Petkov{\v s}ek.
	\newblock Hypergeometric solutions of linear recurrences with polynomial
	coefficients.
	\newblock {\em J.~Symbolic Comput.}, 14(2-3):243--264, 1992.
	
	\bibitem{vanHoeij:99}
	\relax{M.~van} Hoeij.
	\newblock Finite singularities and hypergeometric solutions of linear
	recurrence equations.
	\newblock {\em J.~Pure Appl. Algebra}, 139(1-3):109--131, 1999.
	
	\bibitem{Salvy}
	B.~Salvy and P.~Zimmermann.
	\newblock Gfun: A package for the manipulation of generating and holonomic
	functions in one variable.
	\newblock {\em ACM Transactions on Mathematical Software}, 20(20):163–177,
	1994.
	
	\bibitem{Schneider:05c}
	C.~Schneider.
	\newblock Product representations in ${\Pi}{\Sigma}$-fields.
	\newblock {\em Ann. Comb.}, 9(1):75--99, 2005.
	
	\bibitem{Schneider:05a}
	C.~Schneider.
	\newblock Solving parameterized linear difference equations in terms of
	indefinite nested sums and products.
	\newblock {\em J. Differ. Equ. Appl.}, 11(9):799--821, 2005.
	
	\bibitem{Schneider:07a}
	C.~Schneider.
	\newblock Symbolic summation assists combinatorics.
	\newblock {\em S\'em.~Lothar. Combin.}, 56:1--36, 2007.
	\newblock Article B56b.
	
	\bibitem{Schneider:10c}
	C.~Schneider.
	\newblock Parameterized telescoping proves algebraic independence of sums.
	\newblock {\em Ann. Comb.}, 14:533--552, 2010.
	\newblock arXiv:0808.2596; see also Proc.~FPSAC 2007.
	
	\bibitem{DR1}
	C.~Schneider.
	\newblock A difference ring theory for symbolic summation.
	\newblock {\em J. Symb. Comput.}, 72:82--127, 2016.
	\newblock arXiv:1408.2776.
	
	\bibitem{DR3}
	C.~Schneider.
	\newblock {Summation Theory II: Characterizations of $R\Pi\Sigma$-extensions
		and algorithmic aspects}.
	\newblock {\em J. Symb. Comput.}, 80(3):616--664, 2017.
	\newblock arXiv:1603.04285.
	
	\bibitem{SchneiderProd:20}
	C.~Schneider.
	\newblock {Minimal representations and algebraic relations for single nested
		products}.
	\newblock {\em Progr. and Computer Software}, 46(2):133--161, 2020.
	\newblock arXiv:1911.04837.
	
	\bibitem{Schneider:21}
	C.~Schneider.
	\newblock {Term Algebras, Canonical Representations and Difference Ring Theory
		for Symbolic Summation}.
	\newblock RISC Report Series 21-03, 2021.
	\newblock arXiv:2102.01471.
	
	\bibitem{Singer:97}
	M.~van~der Put and M.~Singer.
	\newblock {\em Galois theory of difference equations}, volume 1666 of {\em
		Lecture Notes in Mathematics}.
	\newblock Springer-Verlag, Berlin, 1997.
	
	\bibitem{Zeilberger:91}
	D.~Zeilberger.
	\newblock The method of creative telescoping.
	\newblock {\em J.~Symb. Comput.}, 11:195--204, 1991.
	
\end{thebibliography}

\end{document}